\documentclass[a4paper]{llncs}

\usepackage{xspace}
\usepackage{graphicx}
\usepackage{algorithm}
\usepackage{algorithmic}
\usepackage{theorem}
\usepackage{amsmath,amssymb}
\usepackage{amsfonts}

\usepackage{color}

\def\calf{\mathcal{f}}

\def\AFO{\textbf{\textsf{AFO}}}

\def\Tr{\textsf{True}}

\def\f{{\theta}}
\def\r{{\omega}}
\def\p{{{\phi}}}
\def\d{{{\delta}}}

\def\Cons{\textsf{\textbf{C}}}

\def\Var{\textsf{\textbf{{V}}}}
\def\PredE{\textsf{\textbf{{P}}}_\textit{{EM}}}
\def\PredA{\textsf{\textbf{{P}}}_\textit{{AM}}}
\def\GndE{\textsf{\textbf{{G}}}_\textit{{EM}}}
\def\GndA{\textsf{\textbf{{G}}}_\textit{{AM}}}
\def\wldE{\mathcal{W}_\textit{{EM}}}
\def\wldA{\mathcal{W}_\textit{{AM}}}
\def\cala{\mathcal{A}}
\def\ICE{\textsf{\textbf{{IC}}}_\textit{{EM}}}
\def\ICA{\textsf{\textbf{{IC}}}_\textit{{AM}}}

\def\Pr{\textsf{\textit{Pr}}}
\def\kE{\mathcal{K}_\textit{{EM}}}
\def\PAM{\Pi_{\textit{AM}}}
\def\PAMI{\Pi_{\textit{AM}}^{\cali}}
\def\PEM{\Pi_{\textit{EM}}}
\def\cali{\mathcal{I}}
\def\fEM{\textit{form}_{EM}}
\def\fAM{\textit{form}_{AM}}
\def\war{\vdash_{\textsf{war}}}
\def\nwar{\not\vdash_{\textsf{war}}}
\def\must{nec}
\def\can{poss}

\def\calf{\mathcal{F}}
\def\calfs{\calf^*}
\def\<{\langle}
\def\>{\rangle}
\def\calt{\mathcal{T}}
\def\calts{\mathcal{T}^*}

\def\af{\textit{af{\,}}}

\def\Pdist{\textbf{\textsf{P}}_{L,\Pr,\cali}}

\def\AfCp{CandPgm_{af}}

\newtheorem{observation}{Observation}
\newcommand{\smallbsq}{\hfill{\tiny $\blacksquare$}}

\begin{document}

\title{Belief Revision in \\
Structured Probabilistic Argumentation}

\author{Paulo Shakarian\inst{1} \and
Gerardo I.\ Simari\inst{2} \and
Marcelo A.\ Falappa\inst{3}}
\institute{
Department of Electrical Engineering and Computer Science \\
U.S.\ Military Academy, West Point, NY, USA \\
\email{paulo@shakarian.net}
\and
Department of Computer Science, University of Oxford, United Kingdom \\
\email{gerardo.simari@cs.ox.ac.uk}
\and
Departamento de Ciencias e Ingenier\'{\i}a de la Computaci\'on \\
Universidad Nacional del Sur, Bah\'{\i}a Blanca, Argentina \\
\email{mfalappa@cs.uns.edu.ar}
}

\maketitle

\begin{abstract}
In real-world applications, knowledge bases consisting of all the information at hand for a specific
domain, along with the current state of affairs, are bound to contain contradictory data coming from
different sources, as well as data with varying degrees of uncertainty attached. Likewise, an important
aspect of the effort associated with maintaining knowledge bases is deciding what information is no
longer useful; pieces of information (such as intelligence reports) may be outdated, may come from sources
that have recently been discovered to be of low quality, or abundant evidence may be available that
contradicts them. In this paper, we propose a probabilistic
structured argumentation framework that arises from the extension of Presumptive Defeasible Logic Programming
(PreDeLP) with probabilistic models, and argue that this formalism is capable of addressing the basic
issues of handling contradictory and uncertain data.
Then, to address the last issue, we focus on the study of non-prioritized belief revision operations
over probabilistic PreDeLP programs. We propose a set of rationality postulates -- based on well-known ones
developed for classical knowledge bases -- that characterize how such operations should behave, and study a
class of operators along with theoretical relationships with the proposed postulates, including a representation
theorem stating the equivalence between this class and the class of operators characterized by the postulates.
\end{abstract}

\newcommand{\from}[3]{\noindent {\bf [{\sc from #1 to #2:} {\small #3}]}}

\newcommand{\SRules}{\mbox{$\Omega$}}
\newcommand{\Facts}{\mbox{$\Theta$}}
\newcommand{\DRules}{\mbox{$\Delta$}}
\newcommand{\Presumps}{\mbox{$\Phi$}}
\newcommand{\pdelpprog}{(\Facts, \SRules, \Presumps,  \DRules)}
\newcommand{\PDeLP}{PreDeLP}

\newcommand{\warrantFormula}{{\sf warrantFormula}}

\newcommand{\ag}{\mathbb{A}}

\newcommand{\ie}{\textit{i.e.},\xspace}

\newcommand{\eg}{\textit{e.g.,}\xspace}

\newcommand{\etal}{\textit{et al}\xspace}

\newcommand{\Attacks}{\mbox{$attacks$}}
\newcommand{\allArgs}{\mbox{$AR$}}

\newcommand{\ArgSB}[2]{$\langle#1,#2\rangle_{_\SB}$}
\newcommand{\Args}[2]{$\langle#1,#2\rangle$}
\newcommand{\AS}[2]{\langle#1,#2\rangle}
\newcommand{\comp}[1]{\mbox{$\overline{#1}$}}
\newcommand{\layer}[1]{\mathit{index}_{_\SB}(#1)}

\newcommand{\ArgFrame}{\mbox{$\langle\allArgs, \Attacks \rangle$}}

\newcommand{\dexp}{$\delta$-Explanation}
\newcommand{\dexps}{\mbox{$\delta$-Explanations}}

\newcommand{\yes}{{\sc yes}}
\newcommand{\ansno}{{\sc no}}
\newcommand{\und}{{\sc undecided}}
\newcommand{\unk}{{\sc unknown}}

\newcommand{\prolog}{{\sc prolog}}
\newcommand{\DLP}{\mbox{{\sc DeLP}}}
\newcommand{\delp}{\DLP-program}

\newcommand{\mergesimb}{$\circ$}
\newcommand{\merge}[2]{\textnormal{#1\mergesimb#2}}

\newcommand{\no}{\mbox{$\neg$}}
\newcommand{\PP}{\PAM}
\newcommand{\DD}{\mbox{$\DRules\cup\Presumps$}}
\newcommand{\SSet}{\mbox{$\Pi$}}
\newcommand{\SD}{\mbox{$(\SSet,\DD)$}}
\newcommand{\Prog}{\mbox{$\PP$}}
\newcommand{\pair}[2]{$(#1,#2)$}
\newcommand{\SyA}{\mbox{\SSet\ $\cup$ \Arg}}

\newcommand{\Argu}[1]{\mbox{$\mathcal{A}_#1$}}
\newcommand{\Brgu}[1]{\mbox{$\mathcal{B}_#1$}}
\newcommand{\Arg}{\mbox{$\mathcal{A}$}}
\newcommand{\AL}{\ensuremath{\langle \Arg,L \rangle}}

\newcommand{\srule}[2]{\mbox{$#1\!\leftarrow#2$}}
\newcommand{\facto}[1]{\mbox{$#1$}}

\newcommand{\defleft}{\mbox{\bf--\hspace{-1pt}\raise.1185pt\hbox{$\prec$} }}
\newcommand{\defleftarrow}{{\raise1.5pt\hbox{\tiny\defleft}}}
\newcommand{\drule}[2]{\mbox{$#1 \;\defleftarrow #2$}}

\newcommand{\Tree}[1]{\mbox{${\mathcal T}({\small #1})$}}
\newcommand{\MTree}[1]{\mbox{${\mathcal T^*}({\small #1})$}}

\newcommand{\Aline}{\mbox{$\Lambda$}}
\newcommand{\Sline}{\mbox{$\Aline_S$}}
\newcommand{\Iline}{\mbox{$\Aline_I$}}

\newcommand{\Ar}[2]{\mbox{$\langle #1,#2 \rangle $}}
\newcommand{\Arbrace}[2]{\mbox{$\langle \{#1\},#2 \rangle $}}
\newcommand{\Aho}{\mbox{$\langle \Arg_0,h_0 \rangle $}}
\newcommand{\Aha}{\mbox{$\langle \Arg_1,h_1 \rangle $}}
\newcommand{\Ahb}{\mbox{$\langle \Arg_2,h_2 \rangle $}}
\newcommand{\Ahc}{\mbox{$\langle \Arg_3,h_3 \rangle $}}
\newcommand{\Ahd}{\mbox{$\langle \Arg_4,h_4 \rangle $}}
\newcommand{\Ahe}{\mbox{$\langle \Arg_5,h_5 \rangle $}}
\newcommand{\Ahi}{\mbox{$\langle \Arg_i,h_i \rangle $}}
\newcommand{\Ahk}{\mbox{$\langle \Arg_k,h_k \rangle $}}
\newcommand{\Aheven}{\mbox{$\langle \Arg_{2k},h_{2k} \rangle $}}
\newcommand{\Ahodd}{\mbox{$\langle \Arg_{2k-1},h_{2k-1} \rangle $}}
\newcommand{\AhiAnt}{\mbox{$\langle \Arg_{i-1},h_{i-1} \rangle $}}
\newcommand{\AhiProx}{\mbox{$\langle \Arg_{i+1},h_{i+1} \rangle $}}
\newcommand{\Ahn}{\mbox{$\langle \Arg_n,h_n \rangle $}}
\newcommand{\Ahdp}{\mbox{$\langle \Arg_4',h_4' \rangle $}}

\newcommand{\Barg}{\mbox{${\mathcal B}$}}
\newcommand{\Carg}{\mbox{${\mathcal C}$}}
\newcommand{\Darg}{\mbox{${\mathcal D}$}}
\newcommand{\Bqa}{\mbox{$\langle \Barg_1,q_1 \rangle $}}
\newcommand{\Bqb}{\mbox{$\langle \Barg_2,q_2 \rangle $}}
\newcommand{\Bqi}{\mbox{$\langle \Barg_i,q_i \rangle $}}
\newcommand{\Bqk}{\mbox{$\langle \Barg_k,q_k \rangle $}}
\newcommand{\Bq}{\mbox{$\langle \Barg,q \rangle $}}
\newcommand{\Ah}{\mbox{$\langle \Arg,h \rangle $}}

\newcommand{\Dnode}{\mbox{``\textsf{D}''}}
\newcommand{\Unode}{\mbox{``\textsf{U}''}}

\newcommand{\expla}[2]{\mbox{$\mathcal{E}_{#1}(#2)$}}
\newcommand{\expl}{\mbox{$\mathcal{E}$}}

\newcommand{\A}{\mathcal{A}}
\newcommand{\B}{\mathcal{B}}
\newcommand{\C}{\mathcal{C}}
\newcommand{\D}{\mathcal{D}}
\newcommand{\E}{\mathcal{E}}
\newcommand{\F}{\mathcal{F}}
\newcommand{\G}{\mathcal{G}}
\renewcommand{\H}{\mathcal{H}}
\newcommand{\I}{\mathcal{I}}
\newcommand{\J}{\mathcal{J}}
\newcommand{\K}{\mathcal{K}}
\renewcommand{\L}{\mathcal{L}}
\newcommand{\M}{\mathcal{M}}
\newcommand{\N}{\mathcal{N}}
\renewcommand{\O}{\mathcal{O}}
\renewcommand{\P}{\mathcal{P}}
\newcommand{\Q}{\mathcal{Q}}
\newcommand{\R}{\mathcal{R}}
\renewcommand{\S}{\mathcal{S}}
\newcommand{\T}{\mathcal{T}}
\newcommand{\U}{\mathcal{U}}
\newcommand{\V}{\mathcal{V}}
\newcommand{\W}{\mathcal{W}}
\newcommand{\X}{\mathcal{X}}
\newcommand{\Y}{\mathcal{Y}}
\newcommand{\Z}{\mathcal{Z}}
\newcommand{\Pre}{\mathcal{P}}

\def\StrPt{\SRules\cup\Facts}
\def\kitSink{\StrPt\cup\DD}


\section{Introduction and Related Work}
\label{sec:intro}

Decision-support systems that are part of virtually any kind of real-world application must be part of
a framework that is rich enough to deal with several basic problems:
(i) handling contradictory information;
(ii) answering abductive queries;
(iii) managing uncertainty; and (iv) updating beliefs.
{\em Presumptions} come into play as key components of answers to abductive queries, and must be maintained
as elements of the knowledge base; therefore, whenever candidate answers to these queries are evaluated,
the (in)consistency of the knowledge base together with the presumptions being made needs to be addressed
via belief revision operations.

In this paper, we begin by proposing a framework that addresses
items (i)--(iii) by extending Presumptive \DLP~\cite{MartinezGS12} (\PDeLP, for short) with probabilistic
models in order to model uncertainty in the application domain; the resulting framework
is a general-purpose probabilistic argumentation language that we will refer to as Probabilistic \PDeLP (P-\PDeLP,
for short).

In the second part of this paper, we address the problem of updating beliefs -- item (iv) above -- in P-\PDeLP\
knowledge bases, focusing on the study of non-prioritized belief revision operations.
We propose a set of rationality postulates characterizing how such operations should behave -- these postulates
are based on the well-known postulates proposed in~\cite{hansson1997semi} for non-prioritized belief revision in
classical knowledge bases. We then study a class of operators and their theoretical relationships with
the proposed postulates, concluding with a representation theorem.


\medskip
\noindent\textbf{Related Work.}  Belief revision studies changes to knowledge bases as a response to
{\em epistemic inputs}. Traditionally,
such knowledge bases can be either belief sets (sets of formulas closed under
consequence)~\cite{AlchourronGM85,gar88} or belief
bases~\cite{Hansson94,hansson1997semi} (which are not closed);
since our end goal is to apply the results we obtain to real-world domains,
here we focus on belief bases. In particular, as motivated by requirements (i)--(iv) above, our knowledge bases
consist of logical formulas over which we apply argumentation-based reasoning and to which we couple a
probabilistic model.
The connection between belief revision and argumentation was first studied
in~\cite{Doyle79}; since then, the work that is most closely related to our approach is
the development of the explanation-based operators of~\cite{FalappaKS02}.

The study of argumentation systems together with probabilistic reasoning has recently
received a lot attention, though a significant part has been in the combination between the two has been in the form
of probabilistic abstract argumentation~\cite{LiON11,Thimm12,Hunter12,FazzingaFP13}. There have, however, been several
approaches that combine structured argumentation with models for reasoning under uncertainty; the first of such
approaches to be proposed was~\cite{haenni1999probarg}, and several others followed, such as
the possibilistic approach of~\cite{Chesnevar:2004}, and the probabilistic logic-based approach
of~\cite{Hunter13}.
The main difference between these works and our own is that here we adopt a bipartite knowledge
base, where one part models the knowledge that is not inherently probabilistic -- uncertain knowledge is modeled
separately, thus allowing a clear separation of interests between the two kinds of models.
This approach is based on a similar one developed for ontological languages in the Semantic Web
(see~\cite{LukSM13}, and references within).


Finally, to the best of our knowledge, this is the first paper in which the combination of structured
argumentation, probabilistic models, and belief revision has been addressed in conjunction.



\section{Preliminaries}
\label{sec:prelims}

The Probabilistic \PDeLP\ (P-\PDeLP, for short) framework is composed of two separate models of the world.
The first is called the {\em environmental model} (referred to as ``EM''), and is used to describe the
probabilistic knowledge that we have about the domain.
The second one is called the {\em analytical model} (referred to as ``AM''), and is used to analyze
competing hypotheses that can account for a given phenomenon -- what we will generally call queries.
The AM is composed of a classical (that is, non-probabilistic) \PDeLP\ program in order to allow for
contradictory information, giving the system the capability to model competing explanations for a given query.

\medskip
\noindent
{\bf Two Kinds of Uncertainty.}
In general, the EM contains knowledge such as evidence, uncertain facts, or knowledge about agents and systems.
The AM, on the other hand, contains ideas that a user may conclude based on the information in the EM.
Table~\ref{emAmTable} gives some examples of the types of information that could appear in each of the two models
in a cyber-security application.
Note that a knowledge engineer (or automated system) could assign a probability to statements in the EM column, whereas
statements in the AM column can be either true or false depending on a certain combination (or several possible
combinations) of statements from the EM.
There are thus two kinds of uncertainty that need to be modeled: probabilistic uncertainty and uncertainty arising
from defeasible knowledge. As we will see, our model allows both kinds of uncertainty to coexist, and also allows
for the combination of the two since defeasible rules and presumptions (that is, defeasible facts) can also be
annotated with probabilistic events.

In the rest of this section, we formally describe these two models, as well as how knowledge in the AM can be
annotated with information from the EM -- these annotations specify the conditions under which the various
statements in the AM can potentially be true.

\begin{table}[t]
\centering
\begin{footnotesize}
\begin{tabular}{l|l}
{\bf Probabilistic Model (EM)}         & {\bf Analytical Model (AM)} \\ \hline\hline
``Malware X was compiled on a system   & ``Malware X was compiled on a system in\\
using the English language.''          &  English-speaking country Y.'' \\ \hline
``County Y and country Z are           & ``Country Y has a motive to launch a \\
currently at war.''                    & cyber-attack against country Z \\ \hline
``Malware W and malware X were created & ``Malware W and malware X are related. \\
in a similar coding style.''           &                                        \\ \hline
\end{tabular}
\end{footnotesize}
\caption{Examples of the kind of information that could be represented in the two different models
in a cyber-security application domain.}
\label{emAmTable}
\end{table}

\medskip
\noindent
{\bf Basic Language.}
We assume sets of variable and constant symbols, denoted with
$\Var$ and $\Cons$, respectively. In the rest of this paper, we will use capital letters to represent variables
(e.g., $X,Y,Z$), while lowercase letters represent constants.
The next component of the language is a set of $n$-ary predicate symbols; the EM and AM use separate sets of predicate
symbols, denoted with $\PredE,\PredA$, respectively -- the two models can, however, share variables and constants.
As usual, a {\em term} is composed of either a variable or constant. Given terms $t_1,..., t_n$ and $n$-ary predicate symbol
$p$, $p(t_1,...,t_n)$ is called an {\em atom}; if $t_1,...,t_n$ are constants, then the atom is said to be {\em ground}.
The sets of all ground atoms for EM and AM are denoted with $\GndE$ and $\GndA$, respectively.

Given set of ground atoms, a \textit{world} is any subset of atoms -- those that belong to the set are said to be
{\em true} in the world, while those that do not are {\em false}.  Therefore, there are $2^{|\GndE|}$ possible worlds in the
EM and $2^{|\GndA|}$ worlds in the AM. These sets are denoted with $\wldE$ and $\wldA$, respectively.
In order to avoid worlds that do not model possible situations given a particular domain, we include {\em integrity
constraints} of the form $\textsf{oneOf}(\cala')$,  where $\cala'$ is a subset of ground atoms.  Intuitively, such a constraint
states that any world where more than one of the atoms from set $\cala'$ appears is invalid.
We use $\ICE$ and $\ICA$ to denote the sets of integrity constraints for the EM and AM,
respectively, and the sets of worlds that conform to these constraints is denoted with $\wldE(\ICE),\wldA(\ICA)$,
respectively.

Finally, logical formulas arise from the combination of atoms using the traditional connectives
($\wedge$, $\vee$, and $\neg$).
As usual, we say a world $w$ {\em satisfies} formula ($f$), written $w\models f$, iff:
(i) If $f$ is an atom, then $w \models f$ iff $f \in w$;
(ii) if $f = \neg f'$ then $w \models f$ iff $w \not\models f'$;
(iii) if $f=f' \wedge f''$ then $w \models f$ iff $w\models f'$ and $w\models f''$; and
(iv) if $f=f' \vee f''$ then $w \models f$ iff $w \models f'$ or $w \models f''$.
We use the notation $\fEM, \fAM$ to denote the set of all possible (ground) formulas in the EM and AM, respectively.


\subsection{Probabilistic Model}
\label{sec:probModel}

The EM or environmental model
is largely based on the probabilistic logic of~\cite{nil86}, which we now briefly review.

\begin{definition}
\label{def:probform}
Let $f$ be a formula over $\PredE$, $\Var$, and $\Cons$, $p \in [0,1]$, and $\epsilon \in [0,\min(p,1-p)]$.
A {\em probabilistic formula} is of the form $f : p \pm \epsilon$.
A set $\kE$ of probabilistic formulas is called a {\em probabilistic knowledge base}.
\end{definition}
In the above definition, the number $\epsilon$ is referred to as an {\em error tolerance}.
Intuitively, probabilistic formulas are interpreted as ``formula $f$ is true with probability between $p - \epsilon$
and $p + \epsilon$'' -- note that there are no further constraints over this interval apart from those imposed by other
probabilistic formulas in the knowledge base.
The uncertainty regarding the probability values stems from the fact
that certain assumptions (such as probabilistic independence) may not be suitable in the environment being modeled.

\begin{example}
\label{ex:EM}
\label{logPgmEx}
Consider the following set $\kE$:

\medskip
$
\begin{array}{llllllllllll}
f_1 &=& a          & : 0.8 \pm 0.1 \ \ \ \ \ & f_4 &=& d \wedge e &: 0.7 \pm 0.2  \ \ \ \ \ & f_7 &=& k                   &: 1 \pm 0\\
f_2 &=& b          & : 0.2 \pm 0.1 & f_5 &=& f \wedge g \wedge h &: 0.6 \pm 0.1 &&&&\\
f_3 &=& c          &: 0.8 \pm 0.1  & f_6 &=& i \vee \neg j       &: 0.9 \pm 0.1 &&&&\\
\end{array}
$

\medskip
\noindent
Throughout the paper, we also use $\kE' =\{f_1,f_2,f_3\}$
\vspace*{-8pt}
\smallbsq
\end{example}

A set of probabilistic formulas describes a set of possible probability distributions $\Pr$ over the set $\wldE(\ICE)$.
We say that probability distribution $\Pr$ \textit{satisfies} probabilistic formula $f:p \pm \epsilon$ iff:
$
p - \epsilon \leq \sum_{w \in \wldE(\ICE)}\Pr(w) \leq p+\epsilon.
$
We say that a probability
distribution over $\wldE(\ICE)$ \emph{satisfies} $\kE$ iff it satisfies all probabilistic formulas in $\kE$.

Given a probabilistic knowledge base and a (non-probabilistic) formula $q$,
the \textit{maximum entailment} problem seeks to identify real numbers $p,\epsilon$ such that all valid probability
distributions $\Pr$ that satisfy $\kE$ also satisfy $q: p \pm \epsilon$, and there does not exist $p', \epsilon'$ s.t.\
$[p-\epsilon,p+\epsilon] \supset [p'-\epsilon',p'+\epsilon']$, where all probability distributions $\Pr$ that satisfy
$\kE$ also satisfy $q : p' \pm \epsilon'$.
In order to solve this problem we must solve the linear program defined below.

\begin{definition}
Given a knowledge base $\kE$ and a formula $q$, we have a variable $x_i$ for each $w_i \in \wldE(\ICE)$.
\begin{itemize}\addtolength{\itemsep}{6pt}
\item For each $f_j : p_j \pm \epsilon_j \in \kE$, there is a constraint of the form:

\begin{quote}
$
p_j - \epsilon_j \leq \sum_{w_i \in \wldE(\ICE) \textit{ s.t.\ } w_i \models f_j}x_i \leq p_j + \epsilon_j.
$
\end{quote}

\item We also have the constraint:
$
\sum_{w_i \in \wldE(\ICE)} x_i = 1.
$

\item\label{objFcn}
The objective is to minimize the function:
$
\sum_{w_i \in \wldE(\ICE) \textit{ s.t.\ } w_i \models q} x_i.
$
\end{itemize}
We use the notation $\textsf{EP-LP-MIN}(\kE,q)$ to refer to the value of the objective function in the
solution to the \textsf{EM-LP-MIN} constraints.
\end{definition}

The next step is to solve the linear program a second time, but instead maximizing the objective function (we shall
refer to this as \textsf{EM-LP-MAX}) -- let $\ell$ and $u$ be the results of these operations, respectively.
In \cite{nil86}, it is shown that $\epsilon = \frac{u-\ell}{2}$ and $p = \ell+\epsilon$ is the solution to the maximum entailment problem.  We note that although the above linear program has an exponential number of variables in the worst
case (i.e., no integrity constraints), the presence of constraints has the potential to greatly reduce this space.
Further, there are also good heuristics (cf.\ \cite{amai07,amai12}) that have been shown to provide highly accurate
approximations with a reduced-size linear program.

\begin{example}
\label{linPrgEx}
Consider KB $\kE'$ from Example~\ref{ex:EM} and a set of ground atoms restricted to those that appear in that program;
we have the following worlds:

\medskip
$
\begin{array}{lllllllllllllll}
w_1 &=& \{ a, b, c \}  \ \ \ & w_2 &=& \{ a, b \} \ \ \ & w_3 &=& \{ a, c \} \ \ \ & w_4 &=& \{ b, c \} \\
w_5 &=& \{ b \}              & w_6 &=& \{ a \}          & w_7 &=& \{ c \}          & w_8 &=& \emptyset   & & & \\
\end{array}
$
\medskip

\noindent
and suppose we wish to compute the probability for formula
$
q = a \vee c
$.
For each formula in $\kE$ we have a constraint, and for each world above
we have a variable.  An objective function is created based on the worlds that satisfy the query formula
(in this case, worlds $w_1, w_2, w_3, w_4, w_6, w_7$).
Solving $\textsf{EP-LP-MAX}(\kE',q)$ and $\textsf{EP-LP-MIN}(\kE',q)$, we obtain the solution $0.9 \pm 0.1$.
\smallbsq
\end{example}


\section{Argumentation Model}
\label{sec:argModel}

For the analytical model (AM), we choose a structured argumentation framework~\cite{RahwanSimari2009} due to
several characteristics that make such frameworks highly applicable to many domains.  Unlike the EM, which describes
probabilistic information about the state of the real world, the AM must allow for competing ideas.  Therefore,
it must be able to represent contradictory information. The algorithmic approach we shall later describe
allows for the creation of \textit{arguments} based on the AM that may ``compete'' with each other to answer a given
query.  In this competition -- known as a \textit{dialectical process} -- one argument may
defeat another based on a \textit{comparison criterion} that determines the prevailing argument.  Resulting from this
process, certain arguments are \textit{warranted} (those that are not \emph{defeated} by other
arguments) thereby providing a suitable explanation for the answer to a given query.

The transparency provided by the system can allow knowledge engineers to identify potentially incorrect input information and
fine-tune the models or, alternatively, collect more information.  In short, argumentation-based reasoning has been studied
as a natural way to manage a set of inconsistent information -- it is the way humans settle disputes. As we will see,
another desirable characteristic of (structured) argumentation frameworks is that, once a conclusion is reached, we are
left with an explanation of how we arrived at it and information about why a given argument is warranted; this is very
important information for users to have. In the following, we first recall the basics of the underlying
argumentation framework used, and then go on to introduce the analytical model (AM).

\subsection{Defeasible Logic Programming with Presumptions (PreDeLP)}
\label{sec:delp}

Defeasible Logic Programming with Presumptions (\PDeLP)~\cite{MartinezGS12} is a formalism combining
logic programming
with defeasible argumentation; it arises as an extension of classical \DLP~\cite{Garcia-Simari04-TPLP}
with the possibility of having presumptions, as described below -- since this capability is useful in many
applications, we adopt this extended version in this paper.
In this section, we briefly recall the basics of \PDeLP;
we refer the reader to~\cite{Garcia-Simari04-TPLP,MartinezGS12} for the complete presentation.

The formalism contains several different constructs: facts, presumptions, strict rules, and defeasible rules.
Facts are statements about the analysis that can always be considered to be true, while presumptions are statements
that may or may not be true.  Strict rules specify logical consequences of a set of facts or presumptions (similar to an implication, though not the same) that must always occur, while defeasible rules specify logical consequences that may be assumed to be true when
no contradicting information is present.  These building blocks are used in the construction of \emph{arguments},
and are part of a \PDeLP\ program, which is a set of facts, strict rules, presumptions, and defeasible rules.
Formally, we use the notation $\PP = \pdelpprog$ to denote a \PDeLP\ program, where $\SRules$ is the
set of strict rules, $\Facts$ is the set of facts, $\DRules$ is the set of defeasible rules, and $\Presumps$ is the set
of presumptions.  In Figure~\ref{fig:gndArgEx}, we provide an example $\PP$.
We now define these constructs formally.

\begin{figure}[t]
\centering
\fbox{
\parbox{0.95\columnwidth}{
\begin{tabular}{lllllllll}
$\Facts:$ & $\f_{1a}=$ & $p$ \ \ \ & $\f_{1b}=$ & $q$ \ \ \ & $\f_2=$ & $r$ &&\\
\multicolumn{9}{c}{\rule{0.9\columnwidth}{0.4pt}} \\
$\SRules:$& $\r_{1a}=$ & $\neg s \leftarrow t$ \ \ \ & $\r_{1b}=$ & $\neg t \leftarrow s$ \ \ \ & $\r_{2a}=$ & $s \leftarrow p, u, r, v$ \ \ \ & $\r_{2b}=$ & $t \leftarrow q, w, x, v$ \\
\multicolumn{9}{c}{\rule{0.9\columnwidth}{0.4pt}} \\
$\Presumps:$ & $\p_1=$&  $y \; \defleftarrow$ \ \ \ & $\p_2=$&$v \; \defleftarrow$ \ \ \ & $\p_3=$&$\neg z \; \defleftarrow$\\
\multicolumn{9}{c}{\rule{0.9\columnwidth}{0.4pt}}  \\
$\DRules:$& $\d_{1a}=$ & $s \; \defleftarrow p$ \ \ \ & $\d_{1b}=$ & $t \; \defleftarrow q$ \ \ \ & $\d_2=$    & $s \; \defleftarrow u$ \ \ \ & $\d_3=$    & $s \; \defleftarrow r, v$ \\
          & $\d_4=$    & $u \; \defleftarrow y$ \ \ \ & $\d_{5a}=$ & $\neg u \; \defleftarrow \neg z$ \ \ \ & $\d_{5b}=$ & $\neg w \; \defleftarrow \neg n$ \\
\end{tabular}
}}
\caption{An example (propositional) argumentation framework.}
\label{fig:gndArgEx}
\end{figure}

\medskip
\noindent
\textbf{Facts} ($\Facts$) are ground literals representing atomic information or its negation, using
strong negation ``\no''.  Note that all of the literals in our framework must be formed with a predicate from the
set $\PredA$.  Note that information in the form of facts cannot be contradicted.  We will use the notation $[\Facts]$ to denote the set of all possible facts.

\medskip
\noindent
\textbf{Strict Rules} ($\SRules$) represent non-defeasible cause-and-effect information that resembles an implication (though the
semantics is different since the contrapositive does not hold) and are of the form
$
\srule{L_0}{L_1, \ldots, L_n}
$,
where  $L_0$ is a ground literal and $\{L_i\}_{i>0}$ is a set of ground literals.  We will use the notation $[\SRules]$ to denote the set of all possible strict rules.

\medskip
\noindent
\textbf{Presumptions} ($\Presumps$) are ground literals of the same form as facts, except that they are not taken as being true
but rather defeasible, which means that they can be contradicted. Presumptions are denoted in the same manner as facts,
except that the symbol $\defleftarrow$ is added.

\medskip
\noindent
\textbf{Defeasible Rules} ($\DRules$) represent tentative knowledge that can be used if nothing can be posed against it.
Just as presumptions are the defeasible counterpart of facts, defeasible rules are the defeasible counterpart of
strict rules. They are of the form
$
\drule{L_0}{L_1, \ldots, L_n}
$,
where  $L_0$ is a ground literal and $\{L_i\}_{i>0}$ is a set of ground literals.  In both strict and
defeasible rules, {\it strong negation} is allowed in the head of rules, and hence may be used to represent contradictory
knowledge.

\bigskip

Even though the above constructs are ground, we allow for schematic versions with variables that are used to represent
sets of ground rules.
We denote variables with strings starting with an uppercase letter.

\medskip
\noindent
{\bf Arguments.}
Given a query in the form of a ground atom, the goal is to derive arguments for and against it's validity --
derivation follows the same mechanism of logic programming~\cite{ll87}. Since rule heads can contain strong
negation, it is possible to defeasibly derive contradictory literals from a program.
For the treatment of contradictory knowledge, \PDeLP\ incorporates a defeasible argumentation formalism
that allows the identification of the pieces of knowledge that are in conflict
and, through the previously mentioned dialectical process, decides which information prevails as warranted.
This dialectical process involves the construction and evaluation
of arguments, building a \emph{dialectical tree} in the process.  Arguments are formally defined next.

\begin{definition}
An \emph{argument} $\langle\Arg, L\rangle$ for a literal $L$ is a pair of the literal and a (possibly empty) set of the EM ($\Arg \subseteq \PP$) that provides a minimal proof for $L$ meeting the following requirements:
(i) $L$ is defeasibly derived from $\Arg$;
(ii) $\StrPt\cup\Arg$ is not contradictory; and
(iii) $\Arg$ is a minimal subset of $\DD$ satisfying 1 and 2, denoted \AL.

Literal $L$ is called the {\em conclusion} supported by the argument, and $\Arg$ is the \emph{support} of the argument. An argument $\langle \B, L\rangle$
is a {\em subargument} of $\langle \A, L'\rangle$ iff $\B \subseteq \A$.  An argument $\langle\A, L\rangle$ is {\em presumptive} iff $\A \cap \Presumps$ is not empty.  We will also use $\SRules(\A) = \A \cap \SRules$, $\Facts(\A) = \A \cap \Facts$, $\DRules(\A) = \A \cap \DRules$, and $\Presumps(\A) = \A \cap \Presumps$.
\end{definition}
Our definition differs slightly from that of~\cite{Simari-Loui92}, where DeLP is introduced, as we include strict
rules and facts as part of arguments -- the reason for this will become clear in Section~\ref{sec:probdelp}.
Arguments for our scenario are shown next.

\begin{example}
\label{ex:args}
Figure~\ref{fig:gndArgsEx} shows example arguments based on the knowledge base from Figure~\ref{fig:gndArgEx}.
Note that $\<\cala_5,u\>$ is a sub-argument of $\<\cala_2,s\>$ and $\<\cala_3,s\>$.
\smallbsq
\end{example}

\begin{figure}[t]
\centering
\fbox{
\parbox{0.85\columnwidth}{
\begin{tabular}{llll}
$\<\cala_1,s\>$& $\cala_1 = \{\f_{1a},\d_{1a}\}$ \ \ \ \ \ & $\<\cala_2,s\>$& $\cala_2 = \{\p_1,\p_2,\d_4,\r_{2a},\f_{1a},\f_2\}$\\
$\<\cala_3,s\>$& $\cala_3 = \{\p_1,\d_2,\d_4\}$ \ \ \ \ \ &  $\<\cala_4,s\>$& $\cala_4 = \{\p_2,\d_3,\f_2\}$\\
$\<\cala_5,u\>$& $\cala_5 = \{\p_1,\d_4\}$ & $\<\cala_6,\neg s\>$& $\cala_6 = \{\d_{1b},\f_{1b},\r_{1a}\}$\\
$\<\cala_7,\neg u\>$& $\cala_7 = \{\p_{3},\d_{5a}\}$ & & \\
\end{tabular}
}}
\caption{Example ground arguments from the framework of Figure~\ref{fig:gndArgEx}.}
\label{fig:gndArgsEx}
\end{figure}

Given an argument $\langle \A_1, L_1\rangle$, counter-arguments are arguments that contradict it. Argument
$\langle \A_2, L_2\rangle$ is said to {\em counterargue} or {\em attack} $\langle \A_1, L_1\rangle$ at a literal $L'$
iff there exists a subargument $\langle \A, L''\rangle$ of $\langle \A_1, L_1\rangle$ such that the set
$\SRules(\A_1) \cup \SRules(\A_2) \cup \Facts(\A_1) \cup \Facts(\A_2) \cup \{L_2,L''\}$ is contradictory.

\begin{example}
\label{ex:attack}
Consider the arguments from Example~\ref{ex:args}. The following are some of the
attack relationships between them:
$\cala_1$, $\cala_2$, $\cala_3$, and $\cala_4$ all attack $\cala_6$;
$\cala_5$ attacks $\cala_7$; and
$\cala_7$ attacks $\cala_2$.
\smallbsq
\end{example}

A {\em proper defeater} of an argument $\langle A, L\rangle$ is a counter-argument
that -- by some criterion -- is considered to be better than $\langle A, L\rangle$; if the two are incomparable according
to this criterion, the counterargument is said to be a {\em blocking} defeater.
An important characteristic of \PDeLP\ is that the argument comparison criterion is modular, and thus the most appropriate criterion for the domain that is being represented can be selected; the default criterion used in classical defeasible logic programming (from which \PDeLP\ is derived) is \emph{generalized specificity}~\cite{jncl2003}, though an extension of this criterion is required for arguments using presumptions~\cite{MartinezGS12}. We briefly recall this criterion next -- the first
definition is for generalized specificity, which is subsequently used in the definition of presumption-enabled
specificity.

\begin{definition}
\label{def:pspecificity}
Let $\Prog = \pdelpprog$ be a \PDeLP\ program and let $\F$ be the set of all literals that
have a defeasible derivation from $\Prog$.
An argument
$\langle \A_1, L_1\rangle$ is \emph{preferred to} $\langle \A_2, L_2\rangle$, denoted with
$\A_1 \succ_{PS} \A_2$ iff:

\smallskip
\noindent
$(1)$ \ For all $H \subseteq \F$, $\SRules(\cala_1)\cup\SRules(\A_2) \cup H$ is non-contradictory: if there is a derivation for $L_1$
from $\SRules(\A_2)\cup \SRules(\A_1)\cup\DRules(\A_1) \cup H$,
and there is no derivation for $L_1$ from $\SRules(\A_1) \cup \SRules(\A_2) \cup H$, then there is a derivation for $L_2$ from
$\SRules(\A_1) \cup \SRules(\A_2) \cup \DRules(\A_2) \cup H$; and

\smallskip
\noindent
$(2)$ \ there is at least one set $H' \subseteq \F$, $\SRules(\cala_1)\cup\SRules(\A_2) \cup H'$ is non-contradictory, such that there is a
derivation for $L_2$ from $\SRules(\cala_1)\cup\SRules(\A_2) \cup H' \cup \DRules(\A_2)$, there is no derivation for $L_2$ from $\SRules(\cala_1)\cup\SRules(\A_2)  \cup H'$, and
there is no derivation for $L_1$ from $\SRules(\cala_1)\cup\SRules(\A_2)  \cup H' \cup \DRules(\A_1)$.
\end{definition}

Intuitively, the principle of specificity says that, in the presence of two conflicting lines of argument about a
proposition, the one that uses more of the available information is more convincing. A classic example involves
a bird, Tweety, and arguments stating that it both flies (because it is a bird) and doesn't fly (because it is a penguin).
The latter argument uses more information about Tweety -- it is more specific -- and is thus the stronger of the two.

\begin{definition}[\cite{MartinezGS12}]
\label{def:MGEPref}
Let $\Prog = \pdelpprog$ be a \PDeLP\ program. An argument
$\langle \A_1, L_1\rangle$ is \emph{preferred to} $\langle \A_2, L_2\rangle$, denoted with
$\A_1 \succ \A_2$ iff any of the following conditions hold:

\smallskip
\noindent
$(1)$ \ $\langle \A_1, L_1\rangle$ and $\langle \A_2, L_2\rangle$ are both factual arguments and
$\langle \A_1, L_1\rangle \succ_{PS} \langle \A_2, L_2\rangle$.

\smallskip
\noindent
$(2)$ \  $\langle \A_1, L_1\rangle$ is a factual argument and $\langle \A_2, L_2\rangle$ is a presumptive argument.

\smallskip
\noindent
$(3)$ \ $\langle \A_1, L_1\rangle$ and $\langle \A_2, L_2\rangle$ are presumptive arguments, and

\smallskip
\noindent
\hspace*{4mm}$(a)$ \ $\Presumps(\A_1) \subsetneq \Presumps(\A_2)$ or, \label{itema}

\smallskip
\noindent
\hspace*{4mm}$(b)$ \ $\Presumps(\A_1) = \Presumps(\A_2)$ and $\langle \A_1, L_1\rangle \succ_{PS} \langle \A_2, L_2\rangle$.\label{itemb}
\end{definition}
Generally, if $\A, \B$ are arguments with rules $X$ and $Y$, resp., and $X \subset Y$, then $\A$
is stronger than $\B$. This also holds when $\A$ and $\B$ use presumptions $P_1$ and $P_2$, resp.,
and $P_1 \subset P_2$.

\begin{example}
The following are some relationships between arguments from Example~\ref{ex:args},
based on Definitions~\ref{def:pspecificity} and \ref{def:MGEPref}.
\begin{quote}
$\cala_1$ and $\cala_6$ are incomparable (blocking defeaters); \\
$\cala_6 \succ \cala_2$, and thus $\cala_6$ defeats $\cala_2$; \\
$\cala_5$ and $\cala_7$ are incomparable (blocking defeaters).
\smallbsq
\end{quote}
\end{example}

A sequence of arguments called an \emph{argumentation line} thus arises from this attack relation,
where each argument defeats its predecessor. To avoid undesirable sequences, which may represent circular
argumentation lines, in \DLP\ an \emph{argumentation line} is \emph{acceptable} if it
satisfies certain constraints (see~\cite{Garcia-Simari04-TPLP}).
A literal $L$ is \emph{warranted} if there exists a non-defeated argument \Arg\ supporting $L$.

Clearly, there can be more than one defeater for a particular argument $\<\A, L\>$. Therefore, many acceptable
argumentation lines could arise from $\<\A, L\>$, leading to a tree structure.
The tree is built from the set of all argumentation lines rooted in the initial argument.
In a dialectical tree, every node (except the root) represents a defeater of its parent, and leaves
correspond to undefeated arguments.
Each path from the root to a leaf corresponds to a different acceptable argumentation line.
A dialectical tree provides a structure for considering all the possible acceptable argumentation lines that
can be generated for deciding whether an argument is defeated. We call this tree \emph{dialectical} because
it represents an exhaustive dialectical\footnote{In the sense of providing reasons for and against a position.} analysis for the argument in its root.
For a given argument $\<\A, L\>$, we denote the corresponding dialectical tree as \Tree{\AL}.

Given a literal $L$ and an argument \AL, in order to decide whether or not a literal $L$ is warranted,
every node in the dialectical tree \Tree{\AL} is recursively marked as \Dnode\ (\emph{defeated}) or
\Unode\ (\emph{undefeated}), obtaining a marked dialectical tree \MTree{\AL} as follows:
\begin{enumerate}
\item  All leaves in \MTree{\AL} are marked as \Unode s, and
\item  Let \Bq\ be an inner node of \MTree{\AL}.
  Then \Bq\ will be marked as \Unode\
      iff every child of \Bq\ is marked as  \Dnode.
  The node \Bq\ will be marked as \Dnode\
    iff it has at least a child marked as \Unode.
\end{enumerate}

Given an argument \AL\ obtained from $\PP$, if the root of \MTree{\AL} is marked as \Unode, then we will
say that \MTree{\Ah} \emph{warrants} $L$ and that $L$ is \emph{warranted} from $\PP$.
(Warranted arguments correspond to those in the grounded extension of a Dung argumentation system \cite{dung95}.)
There is a further requirement when the arguments in the dialectical tree contains presumptions --
the conjunction of all presumptions used in even (respectively, odd) levels of the tree must be consistent.
This can give rise to multiple trees for a given literal, as there can potentially be different arguments that make
contradictory assumptions.

We can then extend the idea of a dialectical tree to a \textit{dialectical forest}.  For a given literal $L$, a dialectical forest $\calf(L)$ consists of the set of dialectical trees for all arguments for $L$.  We shall denote a marked dialectical forest, the set of all marked dialectical trees for arguments for $L$, as $\calfs(L)$.  Hence, for a literal $L$, we say it is \textit{warranted} if there is at least one argument for that literal in the dialectical forest $\calfs(L)$ that is labeled as $\Unode$, \textit{not warranted} if there is at least one argument for the literal $\neg L$ in the dialectical forest $\calfs(\neg L)$ that is labeled as $\Unode$, and \textit{undecided} otherwise.


\section{Probabilistic \PDeLP}
\label{sec:probdelp}

Probabilistic \PDeLP\ arises from the combination of the environmental and analytical models ($\PEM$ and $\PAM$, respectively).
Intuitively, given $\PAM$, every element of $\kitSink$ might only hold in certain worlds in the set
$\wldE$ -- that is, they are subject to probabilistic events.
Therefore, we associate elements of $\kitSink$ with a formula from
$\fEM$.
For instance, we could associate formula $\textit{rainy}$ to fact $\textit{umbrella}$ to state that the latter only holds
when the probabilistic event $\textit{rainy}$ holds; since weather is uncertain in nature, it has been modeled as part of
the EM.

We can then compute the probabilities of subsets of $\kitSink$ using the information contained in
$\PEM$, as we describe shortly.  The notion of an \textit{annotation function}
associates elements of $\kitSink$ with elements of $\fEM$.

\begin{definition}
An {\em annotation function} is any function $\af : \kitSink \rightarrow \fEM$.  We shall use $[\af]$ to denote the set of all annotation functions.
\end{definition}
We will sometimes denote annotation functions as sets of pairs $(f,\textit{af}(f))$ in order to simplify the presentation.
Figure~\ref{fig:annoEx} shows an example of an annotation function for our running example.

\begin{figure}[t]
\small
\centering
\fbox{
\parbox{0.9\columnwidth}{
\begin{tabular}{llll}
$\textit{af}(\f_{1a}) = \textit{af}(\f_{1b})$ &
$= k \vee \big(f \wedge \big( h \vee (e \wedge l)\big)\big)$ \ \ \ \ \  & $\textit{af}(\p_3)$ & $= b$\\
$\textit{af}(\f_2)$ & $= i$ & $\textit{af}(\d_{1a}) = \textit{af}(\d_{1b})$ & $= \Tr$ \\
$\textit{af}(\r_{1a}) = \textit{af}(\r_{1b})$ & $= \Tr$ & $\textit{af}(\d_2)$ & $= \Tr$ \\
$\textit{af}(\r_{2a}) = \textit{af}(\r_{2b})$ & $= \Tr$ & $\textit{af}(\d_3)$ & $= \Tr$  \\
$\textit{af}(\p_1)$ & $= c \vee a$ & $\textit{af}(\d_4)$ & $= \Tr$ \\
$\textit{af}(\p_2)$ & $= f \wedge m$ & $\textit{af}(\d_{5a}) = \textit{af}(\d_{5b})$ & $= \Tr$
\end{tabular}
}}
\caption{Example annotation function.}
\label{fig:annoEx}
\end{figure}

We now have all the components to formally define Probabilistic \PDeLP\ programs (P-\PDeLP\ for short).

\begin{definition}
Given environmental model $\PEM$, analytical model $\PAM$, and annotation function $\af$,
a {\em probabilistic \PDeLP\ program} is of the form
$\cali = (\PEM,\PAM,\af)$.
We use notation $[\cali]$ to denote the set of all possible programs.
\end{definition}

Given this setup, we can consider a world-based approach; that is,
the defeat relationship among arguments depends on the current state of the (EM) world.

\begin{definition}
Let $\cali = (\PEM,\PAM,\af)$ be a P-\PDeLP\ program, argument $\<\A, L\>$ is {\em valid} w.r.t.\ world
$w \in \wldE$ iff $\forall c \in \A, w \models \textit{af}(c)$.
\end{definition}

We extend the notion of validity to argumentation lines, dialectical trees, and dialectical forests in the
expected way (for instance, an argumentation line is valid w.r.t.\ $w$ iff all arguments that comprise that
line are valid w.r.t.\ $w$).
We also extend the idea of a dialectical tree w.r.t.\ worlds; so, for a given world $w \in \wldE$, the
dialectical (resp., marked dialectical) tree induced by $w$ is denoted with $\calt_w{\AL}$ (resp., $\calts_w{\AL}$).
We require that all
arguments and defeaters in these trees to be valid with respect to $w$.  Likewise, we extend the notion of dialectical
forests in the same manner (denoted with $\calf_w(L)$ and $\calfs_w(L)$, resp.).
Based on these concepts\, we introduce the notion of \textit{warranting scenario}.

\begin{definition}
\label{def:warrantScenario}
Let $\cali = (\PEM, \PAM, \af)$ be a P-\PDeLP\ program and
$L$ be a literal formed with a ground atom from $\GndA$;
a world $w \in \wldE$ is said to be a {\em warranting scenario} for $L$
(denoted $w \war L$)
iff there is a dialectical forest $\calfs_w(L)$ in which $L$ is warranted 
and $\calfs_w(L)$ is valid w.r.t.\ $w$.
\end{definition}

Hence, the set of worlds in the EM where a literal $L$ in the AM \textit{must} be true is exactly the set of
warranting scenarios -- these are the ``necessary'' worlds:
$
\must(L) = \{w \in \wldE \; | \; (w \war L)
\}
$.
Now, the set of worlds in the EM where AM literal $L$ \textit{can} be true is the following -- these are the ``possible'' worlds:
$
\can(L) = \{ w \in \wldE \; | \; w \nwar \neg L \}.
$
The probability distribution $\Pr$ defined over the worlds in the EM
induces an upper and lower bound on the probability of literal $L$ (denoted $\Pdist$) as follows:
\[
\ell_{L,\Pr,\cali}= \sum_{w \in \must(L)}\Pr(w), \ \ \ \ \ u_{L,\Pr,\cali} = \sum_{w \in \can(L)}\Pr(w)
\]
\[
\ell_{L,\Pr,\cali} \leq \Pdist \leq  u_{L,\Pr,\cali}
\]
Since the EM in general does not define a single probability distribution, the above computations should be
done using linear programs \textsf{EP-LP-MIN} and \textsf{EP-LP-MAX}, as described above.


\subsection{Sources of Inconsistency}
\label{sec:sources}

We use the following notion of (classical) consistency of \PDeLP\ programs:
$\Pi$ is said to be {\em consistent} if there does not exist ground literal $a$ s.t.\
$\Pi \vdash a$ and $\Pi \vdash \neg a$.
For P-\PDeLP\ programs, there are two main kinds of inconsistency that can be present;
the first is what we refer to as EM, or Type I, (in)consistency.

\begin{definition}
Environmental model $\PEM$ is {\em Type I consistent} iff there exists a probability distribution $\Pr$ over the set of
worlds $\wldE$ that satisfies $\PEM$.
\end{definition}
We illustrate this type of consistency in the following example.

\begin{example}
The following formula is a simple example of an EM for which there is no satisfying probability distribution:
\begin{eqnarray*}
rain &  \vee   & hail : 0.3 \pm 0; \\
rain &  \wedge & hail : 0.5 \pm 0.1.
\end{eqnarray*}
A P-\PDeLP\ program using such an EM gives rise to an example of Type~I inconsistency, as it arises from the fact 
that there is no satisfying interpretation for the EM knowledge base.
\smallbsq
\end{example}

Assuming a consistent EM, inconsistencies can still arise through the interaction between the annotation function and
facts and strict rules. We will refer to this as combined, or Type II, (in)consistency.

\begin{definition}
\label{type2def}
A P-\PDeLP\ program $\cali = (\PEM, \PAM, \af)$, with $\PAM$ $=$ \linebreak
$\<\Theta,\Omega,\Phi,\Delta\>$, is {\em Type II consistent} iff:
given any probability distribution $\Pr$ that satisfies $\PEM$,
if there exists a world $w \in \wldE$ such that
$\bigcup_{x \in \Theta \cup \Omega \, | \, w \models \textit{af}(x)}\{x\}$
is inconsistent, then we have $\Pr(w) = 0$.
\end{definition}
Thus, any EM world in which the set of associated facts and strict rules are inconsistent
(we refer to this as ``classical consistency'') must always be assigned a zero probability.
The following is an example of this other type of inconsistency.

\begin{example}
\label{ex:incon-II}
Consider the EM knowledge base from Example~\ref{logPgmEx}, 
the AM presented in Figure~\ref{fig:gndArgEx} and the 
annotation function from Figure~\ref{fig:annoEx}.  
Suppose the following fact is added to the argumentation model:
\[
\f_{3}= \neg p,
\]
and that the annotation function is expanded as follows:
\[
\af(\f_3)= \neg k.
\]
Clearly, fact $\f_3$ is in direct conflict with fact $\f_{1a}$ -- this does not necessarily mean that there 
is an inconsistency.
For instance, by the annotation function, $\f_{1a}$ holds in the world $\{k\}$ while $\f_3$ does not.
However, if we consider the world:
\[
w=\{f,h)
\]
Note that $w \models \af(\f_3)$ and $w \models \af(\f_2)$, which means that, in this world, two contradictory 
facts can occur. Since the environmental model indicates that this world can be assigned a non-zero probability,
we have a Type~II inconsist program.
\smallbsq
\end{example}
Another example (perhaps easier to visualize) in the rain/hail scenario discussed above, is as follows: 
suppose we have facts $f = umbrella$ and $g = \neg umbrella$, and annotation function 
$\af(f) = rain \vee hail$ and
$\af(g) = wind$. Intuitively, the first fact states that an umbrella should be carried if it either rains or hails, 
while the 
second states that an umbrella should not be carried if it is windy. If the EM assigns a non-zero probability to 
formula $(rain \vee hail) \wedge wind$, then we have Type~II inconsistency.

In the following, we say that a P-\PDeLP\ program is \textbf{consistent} if and only if it is both Type~I and 
Type~II consistent.
However, in this paper, we focus on Type~II consistency and assume that the program is~Type I consistent.

\subsection{Basic Operations for Restoring Consistency}
\label{sec:basicOps}

Given a P-\PDeLP\ program that is Type II inconsistent, there are two basic strategies that can be used to
restore consistency:

\medskip
\noindent
{\em Revise the EM}: the probabilistic model can be changed in order to force the worlds that induce
contradicting strict knowledge to have probability zero.

\smallskip
\noindent
{\em Revise the annotation function}: The annotations involved in the inconsistency can be changed so that
the conflicting information in the AM does not become induced under any possible world.

\medskip
It may also appear that a third option would be to adjust the AM -- this is, however, equivalent to
modifying the annotation function.  Consider the presence of two facts in the AM: $a, \neg a$.
Assuming that this causes an inconsistency (that is, there is at least one world in which they both hold),
one way to resolve it would be to remove one of these two literals. Suppose $\neg a$ is removed;
this would be equivalent to setting $\textit{af}(\neg a)=\bot$ (where $\bot$ represents a contradiction
in the language of the EM).
{\em In this paper, we often refer to ``removing elements of $\PAM$'' to refer to changes to the annotation
function that cause certain elements of the $\PAM$ to not have their annotations satisfied in certain EM worlds.}

Now, suppose that $\PEM$ is consistent, but that the overall program is Type II inconsistent.
Then, there must exist a set of worlds in the EM where there is a probability distribution that assigns each
of them a non-zero probability.
This gives rise to the following result.

\begin{proposition}
\label{prop:zeroOut}
If there exists a probability distribution $\Pr$ that satisfies $\PEM$ s.t.\ there exists a world
$w \in \wldE$ where $\Pr(w)>0$ and
$\bigcup_{x \in \Theta \cup \Omega \, | \,w \models \textit{af}(x)}\{x\}$ is inconsistent (Type II inconsistency), then
any change made in order to resolve this inconsistency by modifying only $\PEM$ yields a new EM $\PEM'$
such that $\big(\bigwedge_{a \in w}a \wedge \bigwedge_{a \notin w}\neg a\big) : 0 \pm 0$ is entailed by $\PEM'$.
\end{proposition}

Proposition~\ref{prop:zeroOut} seems to imply an easy strategy of adding formulas to $\PEM$ causing certain worlds to
have a zero probability.  However, this may lead to Type~I inconsistencies in the resulting model $\PEM'$.
If we are applying an EM-only strategy to resolve inconsistencies, this would then lead to further adjustments
to $\PEM'$ in order to restore Type~I consistency.  However, such changes could potentially lead to Type~II
inconsistency in the overall P-\PDeLP\ program (by either removing elements of $\PEM'$ or loosening probability
bounds of the sentences in $\PEM'$), which would lead to setting more EM worlds to a probability of zero.
It is easy to devise an example of a situation in which the probability mass cannot be accommodated given the constraints
imposed by the AM and EM together -- in such cases, it would be impossible to restore consistency by only 
modifying $\PEM$. We thus arrive at the following observation:

\begin{observation}
\label{noEmOnly}
Given a Type II inconsistent P-\PDeLP\ program, consistency cannot always be restored via modifications to
$\PEM$ alone.
\end{observation}

Therefore, due to this line of reasoning, in this paper we focus our efforts on modifications to the
annotation function only.  However, in the future, we intend to explore belief revision operators that consider both
the annotation function (which, as we saw, captures changes to the AM) along with changes to the EM, as well as combinations
of the two.


\section{Revising Probabilistic \PDeLP\ Programs}
\label{sec:conc}

Given a P-\PDeLP\ program $\cali = (\PEM, \PAM, \af)$, with $\PAM = \kitSink$,
we are interested in solving the problem of incorporating an epistemic input $(f,\af')$ into $\cali$, where
$f$ is either an atom or a rule and $\af'$ is equivalent to $\af$, except for its expansion to include $f$.  For ease of presentation, we assume that $f$ is to be incorporated as a fact or strict rule, since incorporating defeasible knowledge can never lead to inconsistency.  As we are only conducting annotation function revisions, for $\cali = (\PEM, \PAM, \af)$ and input $(f,\af')$ we denote the revision as follows: $\cali \bullet (f,\af') = (\PEM, \PAM', \af'')$ where $\PAM'=\PAM\cup\{f\}$
and $\af''$ is the revised annotation function.

\medskip
\noindent
{\bf Notation.} We use the symbol ``$\bullet$'' to denote the revision operator.  We also slightly abuse
notation for the sake of presentation, as well as introduce notation to convert sets of worlds to/from formulas.
\begin{itemize}\addtolength{\itemsep}{6pt}
\item $\cali \cup (f,\af')$ to denote $\cali' = (\PEM,\PAM \cup \{f\}, \af')$.
\item $(f,\af') \in \cali = (\PAM, \PEM, \af)$ to denote $f \in \PAM$ and $\af=\af'$.
\item $wld(f)=\{w \; | \; w \models f\}$ -- the set of worlds that satisfy formula $f$; and
\item $for(w)=\bigwedge_{a \in w}a \wedge \bigwedge_{a \notin w}\neg a$ -- the formula that has $w$ as its only model.
\item ${\PAMI}(w) = \{f \in \Theta \cup \Omega \; | \; w \models \textit{af}(f)\}$
\item $\wldE^{0}(\cali) = \{w \in \wldE \; | \; \PAMI(w) \; \text{is inconsistent}\}$
\item $\wldE^{I}(\cali) = \{w \in \wldE^{0} \; | \; \exists \Pr \textit{ s.t. } \Pr \models \PEM \wedge \Pr(w) >0\}$
\end{itemize}
Intuitively, $\PAMI(w)$ is the subset of facts and strict rules in $\PAM$
whose annotations are true in EM world $w$.
The set $\wldE^{0}(\cali)$ contains all the EM worlds for a given
program where the corresponding knowledge base in the AM is classically inconsistent and
$\wldE^{I}(\cali)$ is a subset of these that can be assigned a non-zero probability --
the latter are the worlds where inconsistency in the AM can arise.

\subsection{Postulates for Revising the Annotation Function}
\label{sec:postulates-af}

We now analyze the rationality postulates for non-prioritized revision of belief bases first introduced
in~\cite{hansson1997semi} and later generalized in~\cite{FalappaKRS12}, 
in the context of P-\PDeLP\ programs.
These postulates are chosen due to the fact that they are well studied in the literature for non-prioritized
belief revision. 

\medskip
\noindent
{\bf Inclusion:}
For $\cali \bullet (f,\af') = (\PEM, \PAM\cup\{f\}, \af'')$, $\forall g \in \PAM$, $wld\big(\af''(g)\big) \subseteq wld(\af'(g))$.

\smallskip
\noindent
This postulate states that, for any element in the AM, the worlds that satisfy its annotation after the revision are
a subset of the original set of worlds satisfying the annotation for that element.

\medskip
\noindent
{\bf Vacuity:}
If $\cali \cup (f,\af')$ is consistent, then $\cali \bullet (f,\af') = \cali \cup (f,\af')$

\medskip
\noindent
{\bf Consistency Preservation:}
If $\cali$ is consistent, then $\cali \bullet (f,\af')$ is also consistent.

\medskip
\noindent
{\bf Weak Success:}
If $\cali \cup (f,\af')$ is consistent, then $(f,\af')\in \cali\bullet(f,\af')$.

\smallskip
\noindent
Whenever the simple addition of the input doesn't cause inconsistencies to arise, the result will contain the input.

\medskip
\noindent
{\bf Core Retainment:}
For $\cali \bullet (f,\af') = (\PEM, \PAM\cup\{f\}, \af'')$, for each
$w \in \wldE^{I}(\cali\cup(f,\af'))$, we have $X_w=\{ h \in \Theta \cup \Omega \; | \; w \models \af''(h)\}$;
for each $g \in \PAM(w)\setminus X_w$
there exists $Y_w \subseteq X_w \cup \{f\}$ s.t.\
$Y_w$ is consistent and $Y_w \cup \{g\}$ is inconsistent.

\smallskip
\noindent
For a given EM world, if a portion of the associated AM knowledge base is removed by the operator, then there
exists a subset of the remaining knowledge base that is not consistent with the removed element and $f$.

\medskip
\noindent
{\bf Relevance:}
For $\cali \bullet (f,\af') = (\PEM, \PAM\cup\{f\}, \af'')$, for each $w \in \wldE^{I}(\cali\cup(f,\af'))$,
we have $X_w=\{ h \in \Theta\cup\Omega \; | \; w \models\af''(h)\}$; for each $g \in \PAM(w)\setminus X_w$ there exists $Y_w \supseteq X_w \cup \{f\}$ s.t.\
$Y_w$ is consistent and $Y_w \cup \{g\}$ is inconsistent.

\smallskip
\noindent
For a given EM world, if a portion of the associated AM knowledge base is removed by the operator, then there
exists a superset of the remaining knowledge base that is not consistent with the removed element and $f$.

\medskip
\noindent
{\bf Uniformity 1:}
Let $(f, \af'_1), (g, \af'_2)$ be two inputs where $\wldE^{I}(\cali\cup(f,\af'_1))=\wldE^{I}(\cali\cup(g,\af'_2))$; for all $w \in \wldE^{I}(\cali\cup(f,\af'))$ and for all
$X \subseteq \PAM(w)$;
if $\{x \; | \; x \in X\cup\{f\}, w \models \af'_1(x)\}$ is inconsistent
iff $\{x \; | \; x \in X\cup\{g\}, w \models \af'_2(x)\}$ is inconsistent, then for each $h \in \PAM$, we have that:
\[
\{ w \in \wldE^{I}(\cali\cup(f,\af'_1)) \; | \; w\models \af'_1(h) \wedge \neg\af''_1(h) \}=
\]
\[
\{ w \in \wldE^{I}(\cali\cup(g,\af'_2)) \; | \; w\models \af'_2(h) \wedge \neg \af''_2(h) \}.
\]

\smallskip
\noindent
If two inputs result in the same set of EM worlds leading to inconsistencies in an AM knowledge base,
and the consistency between analogous subsets (when joined with the respective input) are the same,
then the models removed from the annotation of a given strict rule or fact are the
same for both inputs.

\medskip
\noindent
{\bf Uniformity 2:}
Let $(f, \af'_1), (g, \af'_2)$ be two inputs where $\wldE^{I}(\cali\cup(f,\af'_1))=\wldE^{I}(\cali\cup(g,\af'_2))$; for all $w \in \wldE^{I}(\cali\cup(f,\af'))$and
for all $X \subseteq \PAM(w)$;
if $\{x \; | \; x \in X\cup\{f\}, w \models \af'_1(x)\}$ is inconsistent
iff $\{x \; | \; x \in X\cup\{g\}, w \models \af'_2(x)\}$ is inconsistent, then
\[
\{ w \in \wldE^{I}(\cali\cup(f,\af'_1)) \; | \; w\models \af'_1(h) \wedge \af''_1(h) \}=
\]
\[
\{ w \in \wldE^{I}(\cali\cup(g,\af'_2)) \; | \; w\models \af'_2(h) \wedge \af''_2(h) \}.
\]

\smallskip
\noindent
If two inputs result in the same set of EM worlds leading to inconsistencies in an AM knowledge base, and
the consistency between analogous subsets (when joined with the respective input) are the same, then
the models retained in the the annotation of a given strict rule or fact are the
same for both inputs.

\medskip
\noindent\textbf{Relationships between Postulates.}  
There are a couple of interesting relationships among the postulates. The first is a sufficient condition for 
Core Retainment to be implied by Relevance.

\begin{proposition}
\label{prop:releThenCr}
Let $\bullet$ be an operator such that $\cali \bullet (f,\af') = (\PEM, \PAM\cup\{f\}, \af'')$, 
where $\forall w \in \wldE^{I}(\cali\cup(f,\af'))$, $\PAM^{\cali\bullet(f,\af')}(w)$ is a maximal consistent 
subset of $\PAM^{\cali\cup(f,\af')}(w)$. If $\bullet$ satisfies Relevance then it also satisfies Core Retainment.
\end{proposition}

Similarly, we can show the equivalence between the two Uniformity postulates under certain conditions.

\begin{proposition}
\label{prop:unifEquiv}
Let $\bullet$ be an operator such that $\cali \bullet (f,\af') = (\PEM, \PAM\cup\{f\}, \af'')$ and $\forall w$,
$\PAM^{\cali\bullet(f,\af')}(w) \subseteq \PAM^{\cali\cup(f,\af')}(w)$. Operator $\bullet$ satisfies 
Uniformity~1 iff it satisfies Uniformity~2.
\end{proposition}

\medskip
Given the results of Propositions~\ref{prop:releThenCr} and~\ref{prop:unifEquiv}, we will not study 
Core Retainment and Uniformity~2 with respect to the construction of a belief revision operator in the next section.

\subsection{An Operator for P-\PDeLP\ Revision}
\label{sec:op}

In this section, we introduce an operator for revising a P-\PDeLP\ program.
As stated earlier, any subset of $\PAM$ associated with a world in $\wldE^{I}(\cali\cup(f,\af'))$
must be modified by the operator in order to remain consistent.  So, for such a world $w$,
we introduce a set of candidate replacement programs for $\PAM(w)$ in order to maintain
consistency and satisfy the Inclusion postulate.
\begin{eqnarray*}
candPgm(w,\cali) &=& \{ \PAM' \; | \; \PAM' \subseteq \PAM(w) \;\text{s.t.}\; \PAM' \;\text{is consistent and}\; \\
&& \nexists \PAM'' \subseteq \PAM(w) \textit{ s.t. } \PAM'' \supset \PAM' \;\text{s.t.}\; \PAM'' \;\text{is consistent}\}
\end{eqnarray*}
Intuitively, $candPgm(w,\cali)$ is the set of maximal consistent subsets of $\PAM(w)$.
Coming back to the rain/hail example presented above, we have:
\begin{example}
\label{ex:candpgm}
Consider the P-\PDeLP\ program $\cali$ presented right after Example~\ref{ex:incon-II}, 
and the following EM knowledge base:
\begin{eqnarray*}
rain \vee   hail & : & 0.5 \pm 0.1; \\
rain \wedge hail & : & 0.3 \pm 0.1; \\
            wind & : & 0.2 \pm 0. 
\end{eqnarray*}
Given this setup, we have, for instance:
\[
candPgm(\{rain,hail,wind\},\cali) = \Big\{
\big\{umbrella\big\},
\big\{\neg umbrella\big\} 
\Big\}.
\]
Intuitively, this means that, since the world where $rain$, $hail$, and $wind$ are all true can be assigned a non-zero 
probability by the EM, we must choose
either $umbrella$ or $\neg umbrella$ in order to recover consistency.
\smallbsq
\end{example}

We now show a series of intermediate results that lead up to the representation theorem (Theorem~\ref{theo:rep}).
First, we show how this set plays a role in showing a necessary and
sufficient requirement for Inclusion and Consistency Preservation to hold together.

\begin{lemma}
\label{inclAndCons}
Given program $\cali$ and input $(f,\af')$, operator $\bullet$ satisfies Inclusion and Consistency
Preservation iff for $\cali\bullet(f,\af')=(\PEM,\PAM,\af'')$, for all $w\in \wldE^{I}(\cali\cup(f,\af'))$,
there exists an element $X \in candPgm(w,\cali\cup(f,\af'))$ s.t.\
$\{h \in \Theta \cup \Omega\cup\{f\} \; | \; w \models \af''(h) \} \subseteq X$.
\end{lemma}

Next, we investigate the role that the set $candPgm$ plays in showing the necessary and sufficient
requirement for satisfying Inclusion, Consistency Preservation, and Relevance all at once.

\begin{lemma}
\label{corRetAndRele}
Given program $\cali$ and input $(f,\af')$, operator $\bullet$ satisfies Inclusion, Consistency Preservation,
and Relevance iff for
$\cali\bullet(f,\af')=(\PEM,\PAM,\af'')$, for all
$w\in \wldE^{I}(\cali\cup(f,\af'))$ we have
$\{h \in \Theta \cup \Omega\cup\{f\} \; | \; w \models \af''(h) \} \in candPgm(w,\cali\cup(f,\af'))$.
\end{lemma}

The last of the intermediate results shows that if there is a consistent program where two inputs
cause inconsistencies to arise in the same way, then for each world the set of
candidate replacement programs (minus the added AM formula) is the same.
This result will be used as a support of the satisfaction of the first Uniformity postulate.

\begin{lemma}
\label{uniPCsame-iff-candPgmIsSame}
Let $\cali=(\PEM,\PAM,\af)$ be a consistent program, $(f_1,\af'_1)$, $(f_2,\af'_2)$ be two inputs, and
$\cali_i = (\PEM,\PAM\cup\{f_i\},\af'_i)$.
If
$\wldE^{I}(\cali_1)=\wldE^{I}(\cali_2)$, then for all $w \in \wldE^{I}(\cali_1)$ and all $X \subseteq \PAM(w)$
we have that:
\begin{enumerate}\addtolength{\itemsep}{6pt}
\item If $\{x \; | \; x \in X\cup\{f_1\}, w \models \af'_1(x)\}$ is inconsistent
$\Leftrightarrow$ $\{x \; | \; x \in X\cup\{f_2\}, w \models \af'_2(x)\}$ is inconsistent, then
$\{X \setminus \{f_1\} \; | \; X \in candPgm(w,\cali_1)\}=
\{X \setminus \{f_2\} \; | \; X \in candPgm(w,\cali_2)\}$.
\item If $\{X \setminus \{f_1\} \; | \; X \in candPgm(w,\cali_1)\}=
\{X \setminus \{f_2\} \; | \; X \in candPgm(w,\cali_2)\}$ then $\{x \; | \; x \in X\cup\{f_1\}, w \models \af'_1(x)\}$ is inconsistent $\Leftrightarrow$ $\{x \; | \; x \in X\cup\{f_2\}, w \models \af'_2(x)\}$ is inconsistent.
\end{enumerate}
\end{lemma}

We now have the necessary tools to present the construction of our non-prioritized belief revision operator.

\medskip
\noindent
{\bf Construction.}
Before introducing the construction, we define some preliminary notation. 
Let $\Phi:\wldE \rightarrow 2^{[\Facts]\cup[\SRules]}$.  
For each $h$ there is a formula in $\PAM\cup\{f\}$, where $f$ is part of the input. 
Given these elements, we define:
\begin{eqnarray*}
\textit{newFor}(h,\Phi,\cali,(f,\af')) &=&\af'(h)\wedge\bigwedge_{w \in \wldE^{I}(\cali\cup(f,\af')) \; | \; h \notin \Phi(w)}\neg for(w_i)
\end{eqnarray*}
The following definition then characterizes the class of operators called $\AFO$ (annotation function-based operators).  
\begin{definition}[AF-based Operators]
A belief revision operator $\bullet$ is an ``annotation function-based'' (or af-based) operator
($\bullet \in \AFO$) iff
given program $\cali=(\PEM,\PAM,\af)$ and input $(f,\af')$, the revision is defined as
$\cali\bullet(f,\af')=(\PEM,\PAM\cup\{f\},\af'')$, where:
\[
\forall h, \af''(h) = \textit{newFor}(h,\Phi,\cali,(f,\af'))
\]
where $\forall w\in \wldE$, $\Phi(w)\in\AfCp(w,\cali\cup(f,\af'))$.
\end{definition}

As the main result of the paper, we now show that 
satisfying a key set of postulates is a necessary and sufficient condition for
membership in $\AFO$.

\begin{theorem}[Representation Theorem]
\label{theo:rep}
An operator $\bullet$ belongs to class $\AFO$ iff it satisfies
Inclusion, Vacuity, Consistency Preservation, Weak Success, Relevance,
and Uniformity~1.
\end{theorem}
\begin{proof} (Sketch)
\noindent (If) By the fact that formulas associated with worlds in the set $\wldE^{I}(\cali\cup(f,\af'))$
are considered in the change of the annotation function, Vacuity and Weak Success follow trivially.
Further, Lemma~\ref{corRetAndRele} shows that Inclusion, Consistency Preservation,
and Relevance are satisfied while Lemma~\ref{uniPCsame-iff-candPgmIsSame} shows that
Uniformity~1 is satisfied.

\medskip
\noindent (Only-If) Suppose BWOC that an operator $\bullet$ satisfies all postulates and $\bullet \notin \AFO$.
Then, one of four conditions must hold: (i) it does not satisfy Lemma~\ref{corRetAndRele} or (ii) it does not satisfy Lemma~\ref{uniPCsame-iff-candPgmIsSame}.
However, by those previous arguments, if it satisfies all postulates, these arguments must be true as well
-- hence a contradiction.
$\hfill\Box$
\end{proof}

\section{Conclusions}
\label{sec:conc}

We have proposed an extension of the \PDeLP\ language that allows sentences to be annotated with probabilistic events;
such events are connected to a probabilistic model, allowing a clear separation of interests between
certain and uncertain knowledge. After presenting the language, we focused on characterizing belief revision
operations over P-\PDeLP\ KBs. We presented a set of postulates inspired in the ones presented for non-prioritized
revision of classical belief bases, and then proceeded to study a construction based on these postulates and prove that the
two characterizations are equivalent.

As future work, we plan to study other kinds of operators, such as more general ones that allow the modification of the
EM, as well as others that operate at different levels of granularity. Finally, we are studying the application of P-\PDeLP\
to real-world problems in cyber security and cyber warfare domains.


\medskip
\noindent
{\bf Acknowledgments.}
The authors are partially supported by
UK EPSRC grant EP/J008346/1 (``PrOQAW''),
ERC grant 246858 (``DIADEM''), ARO project 2GDATXR042, DARPA project R.0004972.001,
Consejo Nacional de Investigaciones Cient\'{\i}ficas y T\'ecnicas (CONICET) and
Universidad Nacional del Sur (Argentina).

The opinions in this paper are those of the authors and do
not necessarily reflect the opinions of the funders, the U.S. Military
Academy, or the U.S.\ Army.


\bibliographystyle{splncs}
\bibliography{ProbPreDeLP}  


\end{document}